%% file: main.tex
\documentclass[journal]{IEEEtran}
\IEEEoverridecommandlockouts

\usepackage{cite}
\usepackage{amsmath,amsthm,amssymb,amsfonts}
\usepackage{algorithmic}
\usepackage{graphicx}
 \usepackage{fixltx2e} 
\usepackage{textcomp}
\usepackage{xcolor}
\usepackage{hyperref}
\usepackage{soul}
\usepackage{enumitem}
\usepackage{amsmath}
\usepackage{listings}
\usepackage{xcolor}
\usepackage{comment}
\usepackage{colortbl}         %
\usepackage[table]{xcolor}    %
\usepackage{arydshln}         %
\usepackage{booktabs}
\usepackage{dblfloatfix} %

\lstdefinelanguage{Solidity}{
    keywords={function, returns, private, pure, memory, bytes32, uint256, return, new, for, if, else},
    sensitive=true,
    comment=[l]{//},
    morecomment=[s]{/*}{*/},
    morestring=[b]",
    morestring=[b]',
    morestring=[b]`,
}
\newtheorem{theorem}{Theorem}

\newtheorem{definition}{Definition}

\newcommand{\Adv}{\mathcal{A}}

\def\BibTeX{{\rm B\kern-.05em{\sc i\kern-.025em b}\kern-.08em
    T\kern-.1667em\lower.7ex\hbox{E}\kern-.125emX}}
\begin{document}

\title{Commit-Reveal²: Securing Randomness Beacons with Randomized Reveal Order in Smart Contracts
}

\author{%
  Suhyeon Lee, Euisin Gee, Najmeh Soroush, Muhammed Ali Bingol and Kaibin Huang%
  \thanks{Suhyeon Lee, Euisin Gee, Najmeh Soroush, and Muhammed Ali Bingol are with Tokamak Network (email: \{suhyeon, justin, nil, muhammed\} [at] tokamak.network).}%
}

\maketitle

\begin{abstract}

Simple commit--reveal beacons are vulnerable to last-revealer strategies, and existing descriptions often leave accountability and recovery mechanisms unspecified for practical deployments. We present \emph{Commit--Reveal\textsuperscript{2}}, a layered design for blockchain deployments that cryptographically randomizes the final reveal order, together with a concrete accountability and fallback mechanism that we implement as smart-contract logic. The protocol is architected as a hybrid system, where routine coordination runs off chain for efficiency and the blockchain acts as the trust anchor for commitments and the final arbiter for disputes. Our implementation covers leader coordination, on-chain verification, slashing for non-cooperation, and an explicit on-chain recovery path that maintains progress when off-chain coordination fails. We formally define two security goals for distributed randomness beacons, unpredictability and bit-wise bias resistance, and we show that Commit--Reveal\textsuperscript{2} meets these notions under standard hash assumptions in the random-oracle model. In measurements with small to moderate operator sets, the hybrid design reduces on-chain gas by more than 80\% compared to a fully on-chain baseline. We release a publicly verifiable prototype and evaluation artifacts to support replication and adoption in blockchain applications.

\end{abstract}


\input{1.introduction}

\input{2.related}

\input{3.threat_model}

\input{4.protocol}

\input{5.accountability}

\input{6.security_analysis}

\input{7.implementation}

\input{8.discussion}

\input{9.conclusion}

\section*{Acknowledgement}
The authors would like to thank Aryan Soni from Tokamak Netowrk for implementing the backend that interacted with the smart contract and assisting with functional testing.

\bibliographystyle{IEEEtran}
\bibliography{references}

\input{appendix}

\end{document}

%% file: 1.introduction.tex
\section{Introduction} \label{sec:intro}

Secure randomness generation plays a critical role in blockchain systems. Randomness is essential for validator selection in Proof-of-Stake (PoS) blockchains, cryptographic protocols like Zero-Knowledge Proofs (ZKP), and determining fair execution orders in decentralized finance (DeFi). To systematically provide secure and unbiased randomness, Distributed Randomness Beacons (DRBs) have been proposed as foundational components in blockchain systems. A DRB is a mechanism that generates randomness in a decentralized manner, ensuring that no single participant can control or predict the output. By leveraging multiple participants and cryptographic techniques, DRBs aim to deliver randomness that is both unpredictable and tamper-resistant, which is crucial for maintaining the integrity of various blockchain applications.

However, securely generating randomness in a distributed environment is highly challenging. Randomness generation mechanisms must simultaneously satisfy safety, which ensures resistance to manipulation, and liveness, which guarantees availability when requested. Achieving both properties perfectly is impractical; instead, the degree to which each property is satisfied critically impacts the security and efficiency of blockchain systems.

Among these mechanisms, the Commit-Reveal mechanism is widely adopted due to its structural simplicity and ease of implementation. While it provides strong safety guarantees, it suffers from a critical weakness in liveness: if any participant fails to reveal their secret, randomness generation halts. The \textit{Last Revealer Attack} exacerbates this issue, as the last participant can strategically decide whether to reveal in order to generate a random number or not. This attack becomes particularly concerning when the sequence or outcome determined by the randomness directly impacts financial incentives, giving attackers strong economic motivation.

To address these issues, this paper introduces the \textit{Commit-Reveal²} mechanism, an efficient and secure Distributed Randomness Beacon that can be directly implemented via smart contracts. Commit-Reveal² minimizes opportunities for last revealer attacks by randomizing the reveal order of secret values. Specifically, the mechanism layers two Commit-Reveal processes: the first phase generates randomness used to determine the reveal order in the second phase. This approach reduces the attacker’s ability to manipulate the order while maintaining the same level of security with lower security deposits. Additionally, we introduce an overlapped-commitment structure and optimize communication costs by offloading key reveal processes to off-chain networks. Our primary contributions are:

\begin{itemize}
  \item \textbf{Layered Commit--Reveal$^{2}$.}
  We design a two-layer commit--reveal that cryptographically randomizes the final reveal order. Under standard hash assumptions, this makes it computationally hard for an adversary to ensure a last-revealer position ex ante, thereby weakening the core leverage of the attack. %

  \item \textbf{Gas-Efficient Hybrid Implementation.}
  We architect the protocol as a hybrid on-chain/off-chain system: routine steps occur off chain for efficiency, while the chain serves as a trust anchor for commitments and as the arbiter for disputes. In our measurements for small to moderate operator sets, this reduces gas usage by over 80\% compared to a fully on-chain baseline. %

  \item \textbf{Accountability and Fallback.}
  We implement an on-chain framework that adjudicates participant withholding and leader failure, and offers refunds to consumers during pauses, preserving auditability and aiming for conditional liveness. %

  \item \textbf{Formal Security Framework.}
  We formalize \emph{unpredictability} and \emph{bit-wise bias-resistance} for DRBs and show that Commit--Reveal$^{2}$ satisfies these notions under the random-oracle model. %

  \item \textbf{Public, Verifiable Prototype.}
  We release a complete prototype of Commit--Reveal$^{2}$ and its accountability logic, together with the code and datasets used in our evaluation, to facilitate replication and adoption. %
\end{itemize}

The remainder of this paper is structured to unfold our design. Section~\ref{sec:related} surveys related work. Section~\ref{sec:threat model} defines our security model. Section~\ref{sec:protocol} details the core Commit-Reveal\textsuperscript{2} protocol. Section~\ref{sec:fallback} elaborates on the accountability and fallback mechanisms for the practical use of off-chain network. Section~\ref{sec:security-analysis} and Section~\ref{sec:implementation} present the security analysis and implementation, respectively. Section~\ref{sec:cost analysis} analyzes the overall cost of the proposed protocol and the practicality of griefing attacks. Section~\ref{sec:discussion} discusses extensions, and Section~\ref{section:conclusion} concludes the paper.

%% file: 2.related.tex
\section{Related Works} \label{sec:related}

DRB are crucial primitives in a wide range of blockchain applications, including validator selection in Proof-of-Stake (PoS) protocols, cryptographic schemes such as Zero-Knowledge Protocols (ZKPs), and decentralized finance (DeFi). Despite extensive research, existing solutions often involve trade-offs among unbiasability, liveness, cost-efficiency, and scalability.

\textbf{Commit-Reveal and VDF-Based Schemes.}  
One class of DRBs relies on commit-reveal protocols, where participants first commit to a secret and only reveal it after other commitments are sealed. Simple implementations, such as RANDAO~\cite{randao2019}, provide conceptual clarity and low complexity, but suffer from poor liveness or manipulation~\cite{alpturer2024optimal}, notably the \textit{last revealer attack} where the final participant strategically withholds their secret. To mitigate such issues and ensure unpredictability over time, Verifiable Delay Functions (VDFs)~\cite{boneh2018verifiable} have been integrated with commit-reveal structures, allowing the beacon to \textit{recover} from incomplete reveals and maintain fairness. Protocols like Bicorn~\cite{choi2023bicorn}, which incorporate VDFs to prevent adversarial grinding, illustrate the potential of combining commit-reveal approaches with cryptographic delay primitives. Such hybrid solutions leverage the enforced temporal structure of VDFs to improve both safety and liveness without excessively relying on costly on-chain incentives.

\textbf{VRF and Threshold Cryptography-Based Schemes.}  
Another family of DRB solutions employs Verifiable Random Functions (VRFs) and threshold cryptographic techniques to achieve unbiasable, non-interactive randomness. VRF-based constructions, as utilized in Algorand~\cite{gilad2017algorand} and Ouroboros~\cite{kiayias2017ouroboros}, enable any party to produce a publicly verifiable random output, ensuring integrity and scalability. Threshold-based schemes like RandHound~\cite{syta2017randhound} and HydRand~\cite{schindler2020hydrand} rely on Publicly Verifiable Secret Sharing (PVSS) to distribute trust and ensure robustness against colluding adversaries. By spreading control over multiple participants, these protocols achieve high fault tolerance and transparency, although they may still face communication overhead and incentive alignment issues.

Despite these diverse approaches, research on DRB as smart contracts is limited. Even though efficient DRB designs and implementation in smart contract levels \cite{choi2023bicorn, abidha2024gas, barakbayeva2024srng, lee2024implementation} are proposed, there is a lack of code disclosure, and economic structure analysis for rational use in DApp using smart contracts. Also, these works are all related to DRBs using VDF, which is not considered to be in a practical level considering the gas cost and ASIC chips provision. Our proposed mechanism, \emph{Commit-Reveal$^2$}, aims to address these limitations by combining a layered commit-reveal process with randomized reveal ordering.

%% file: 3.threat_model.tex
\section{Security Model}
\label{sec:threat model}

This section states the threat model for our distributed randomness beacon and fixes the assumptions under which the protocol is analyzed. 

\subsection{System Overview and Assumptions}

We consider a layered Commit--Reveal$^{2}$ design that operates off chain in the common case and shifts to an on-chain fallback when coordination fails, with a designated leader coordinating routine steps. Parties include operators (one of whom is the leader), consumers who request randomness, and the on-chain contract that enforces verification and penalties. Adversaries may control any subset of operators, including the leader, and can withhold messages off chain, submit late within on-chain windows, or attempt to post inconsistent data; all such behavior is adjudicated by the fallback logic on chain.

We assume standard cryptographic and communication primitives. Hash functions are modeled in the random-oracle style in our proofs, commitments and messages are authenticated using EIP-712 signatures and verified by the contract, and off-chain channels between the leader and operators are authenticated. At least one operator can be honest in a round, and deposits with slashing provide accountability when faults are detected. Consumers should be protected by a refund option while service is paused. These assumptions reflect how the hybrid protocol anchors security on chain while minimizing routine on-chain work.

\subsection{Security Goals and Definitions}

With the setting fixed, we now evaluate security along two axes: 
\emph{unpredictability} of future outputs, and \emph{bias-resistance}, formalized as bit-wise unpredictability. 

Let $\lambda \in \mathbb{N}$ be the security parameter. 
A probabilistic polynomial-time (PPT) adversary $\mathcal{A}$ observes the public transcript and past beacon outputs 
$\Omega_1,\ldots,\Omega_{\tau-1}$, and may maintain arbitrary state $st_{\tau-1}$. 
Unless otherwise noted, probabilities are taken over the randomness of the beacon and of the adversary.
We adopt the standard formulations in \cite{meng2024rondo,raikwar2022sok}.

\begin{definition}[Unpredictability]\label{def:unpred}
A DRB is \emph{unpredictable} if for every PPT $\mathcal{A}$,
\begin{equation*}
  \Adv^{\mathrm{unpred}}_{\mathrm{DRB},\mathcal{A}}
  ~\stackrel{\mathrm{def}}{=}~
  \Pr\!\left[
    \begin{aligned}
      &(j,\Omega'_j)\leftarrow \mathcal{A}(\Omega_1,\ldots,\Omega_{\tau-1}, st_{\tau-1});\\
      &j \ge \tau \ \wedge\  \Omega'_j = \Omega_j
    \end{aligned}
  \right]
  \le \operatorname{negl}(\lambda).
\end{equation*}
\end{definition}

Intuitively, bias-resistance requires that no efficient adversary can predict any designated bit of the next output with probability non-negligibly better than $1/2$. 
Let $\Omega_{\tau}\in\{0,1\}^b$ denote the beacon output in round $\tau$ with bit-length $b$, 
and let $[b]\triangleq\{1,\ldots,b\}$.

\begin{definition}[Bias-Resistance (Bit-wise Unpredictability)]\label{def:bias}
A DRB is \emph{bias-resistant} if for every PPT family $\{\mathcal{A}_k\}_{k\in[b]}$,
for every round $\tau\ge 1$ and every position $k\in[b]$,
\begin{align*}
\Adv^{\mathrm{bias}}_{\mathrm{DRB},\mathcal{A}_k}
&= \Pr\!\Big[
  \Omega'_{\tau}[k] \leftarrow \mathcal{A}_k(\Omega_1,\ldots,\Omega_{\tau-1}, st_{\tau-1})
\\[-2pt] &\qquad\land\ \Omega'_{\tau}[k] = \Omega_{\tau}[k]
\Big]
\\
&\le \tfrac{1}{2} + \operatorname{negl}(\lambda).
\end{align*}

\end{definition}

Definition~\ref{def:bias} articulates bit-wise next-bit unpredictability conditioned on the public transcript, thereby ruling out any PPT strategy that skews the distribution of individual output bits by more than a negligible margin. In Definition~\ref{def:unpred}, allowing the adversary to target any future round $j\!\ge\!\tau$ merely strengthens the quantification without changing its essence; specializing to the immediate next round ($j\!=\!\tau$) is equivalent for our analysis. The proofs that follow instantiate these notions under standard hash-function assumptions in the random-oracle model.

%% file: 4.protocol.tex
\section{Protocol Description} \label{sec:protocol}

In this section, we provide a comprehensive overview of the proposed Commit-Reveal\(^2\) protocol. We begin by describing the baseline \emph{on-chain} implementation, which executes all steps directly on the blockchain. Next, we present an \emph{off-chain leveraged} hybrid protocol designed to minimize on-chain costs by optimally shifting storage operations off-chain, requiring only two transactions.

\input{protocol_description}

\subsection{On-Chain Commit-Reveal\texorpdfstring{$^2$}{}}
\label{subsec:on-chain-flow}

Before participating, a node must meet a deposit requirement enforced by the on-chain contract, ensuring that economic incentives are in place to deter malicious behavior. Once a participant has deposited the necessary funds, the Commit-Reveal\(^2\) sequence proceeds entirely on-chain:

\begin{enumerate}
    \item \textbf{Commit Phase:}
    \begin{itemize}
        \item Each participant \(i\) locally generates a secret \(s_i\) from a high-entropy source.
        \item Two layers of commitments are computed:
        \[
            c_{o,i} = \mathrm{Hash}(s_i), 
            \quad
            c_{v,i} = \mathrm{Hash}(c_{o,i}).
        \]
        \item Each participant calls \(\texttt{submit}(c_{v,i})\) on-chain, thereby storing the outer commitment for later verification.
    \end{itemize}

    \item \textbf{Reveal-1 Phase:}
    \begin{itemize}
        \item Participants disclose their first-layer commitment \(c_{o,i}\) by calling \(\texttt{submit}(c_{o,i})\).
        \item The smart contract checks \(\mathrm{Hash}(c_{o,i}) = c_{v,i}\). A successful match means the commitment is revealed correctly.
        \item Once all \(\{c_{o,i}\}\) are revealed, the contract computes 
        \[
            \Omega_v = \mathrm{Hash}(c_{o,1} \,\|\, c_{o,2} \,\|\, \dots \,\|\, c_{o,n}).
        \]
        \item This intermediate value \(\Omega_v\) is used to calculate values of $d_i$ for all participating nodes 
        \[
            d_i = \mathrm{Hash}(\Omega_v \| c_{v,i}).
        \] 
        \item Based on the sequence $d_i$, the \emph{order} in which participants will reveal their final secrets is decided. Nodes are sorted in descending order of $d_i$, meaning a node with a larger value of $d_i$ must reveal its final secret before nodes with smaller values of $d_i$.
    \end{itemize}

    \item \textbf{Reveal-2 Phase:}
    \begin{itemize}
         \item  The full reveal order array is submitted; the contract verifies that \(d_{i-1} > d_i\) holds for each adjacent pair, then saves it.
        \item Following the reveal order, each participant discloses \(s_i\).
        \item The contract confirms \(\mathrm{Hash}(s_i) = c_{o,i}\).
        \item The final randomness is then produced by hashing all revealed secrets:
        \[
            \Omega_o = \mathrm{Hash}(s_1 \,\|\, s_2 \,\|\, \dots \,\|\, s_n).
        \]
    \end{itemize}
\end{enumerate}

By chaining two commit-reveal sequences, participants cannot trivially position themselves last to manipulate the outcome, thus reducing the risk of last-revealer attacks. The final random value \(\Omega_o\) is stored on-chain, ready for consumption by DApps or external services.

\subsection{Off-Chain Leveraged (Hybrid) Commit-Reveal\texorpdfstring{$^2$}{}}
\label{subsec:off-chain-flow}

The hybrid model shifts most storage-related tasks off-chain, reducing on-chain gas costs while preserving key security guaranties via cryptographic proofs.

We assume that all communication between the leader node and participants occurs over a secure, authenticated channel—a standard assumption commonly employed in blockchain protocols that integrate on-chain and off-chain communication, which is achievable with existing standard internet protocols.

\subsubsection{Preparation (Registration and Activation)}

While a deposit is required to participate in both on-chain and hybrid protocols, the hybrid approach employs an off-chain registration and activation step with a leader node. This deposit and registration need only be done once per participant, since the same deposit can be reused in subsequent randomness rounds. Specifically, each node checks the on-chain threshold and provides the necessary deposit. It then submits its Ethereum address, signed authentication, and network details (IP address, port) to the leader node. After verifying the deposit and signature, the leader node updates its local records to formally recognize the node as active.

\subsubsection{Off-Chain Commit and Reveal-1 Phases}

\begin{itemize}
    \item \textbf{Commit Off-Chain:}
    \begin{itemize}
        \item Each node \(i\) locally generates \(s_i\) and calculates \((c_{o,i}, c_{v,i})\).
      \item Rather than individually posting $\{c_{v,i}\}$ on-chain, nodes digitally sign $(\texttt{chainId}, \texttt{verContract}, \texttt{round}, \texttt{attemptId}, c_{v,i})$ under a domain-bound scheme (for example, EIP-712 typed data in Ethereum). Here, \texttt{round} denotes the unique identifier of the randomness round, and \texttt{attemptId} is a per-round counter that increments if the round needs to be retried (e.g., due to a dispute or timeout). The fields \texttt{chainId} and \texttt{verContract} bind signatures to the intended blockchain and contract instance, preventing cross-chain or cross-contract replays. The tuple $(c_{v,i}, \textit{sign}_i)$ is then transferred to the leader node.
        \item The leader node builds a Merkle tree (or an equivalent structure) over \(\{c_{v,i}\}\) and submits the resulting \(\mathrm{MerkleRoot}\) on-chain, thereby storing only minimal data.
    \end{itemize}

    \item \textbf{Reveal-1 Off-Chain:}
    \begin{itemize}
        \item Nodes reveal \(c_{o,i}\) off-chain to the leader node, which verifies \(\mathrm{Hash}(c_{o,i}) = c_{v,i}\).
        \item After confirming all \(\{c_{o,i}\}\), the nodes compute
        \[
            \Omega_v = \mathrm{Hash}(c_{o,1} \,\|\, \dots \,\|\, c_{o,n}),
        \]
         \[
            d_i = \mathrm{Hash}(\Omega_v \| c_{v,i}).
        \] 
    \end{itemize}
\end{itemize}

\subsubsection{Off-Chain Reveal-2 Phase}

\begin{itemize}
    \item \textbf{Secret Reveal:}
    Participants disclose \(\{s_i\}\) off-chain in the order derived from $d_i$. The leader node verifies \(\mathrm{Hash}(s_i) = c_{o,i}\).

\item \textbf{Final Submission On-Chain:}
After collecting all valid secrets and associated signatures, the leader node sends 
$(s_i, c_{o,i}, c_{v,i}, \textit{sign}_i)$ for each participant to the smart contract. 
The contract verifies each signature against the active tuple 
$(\texttt{chainId}, \texttt{verContract}, \texttt{round}, \texttt{attemptId})$, 
ensuring domain separation across chains, contracts, and rounds. 
It then validates each reveal and constructs
\[
    \Omega_o = \mathrm{Hash}(s_1 \,\|\, s_2 \,\|\, \dots \,\|\, s_n),
\]
storing the resulting random value on-chain once verification is complete.
\end{itemize}

By offloading most commit and reveal processes to off-chain channels, the hybrid approach significantly reduces on-chain storage overhead while retaining cryptographic security guaranties. Essential data and final verification are still anchored on-chain, ensuring trust and auditability. This balance between off-chain coordination and on-chain confirmation offers a cost-effective randomness solution. However, its robustness is ultimately guaranteed by a comprehensive set of accountability mechanisms, which we detail in the following section.

%% file: protocol_description.tex
\newcommand{\cmdGen}{\operatorname{Gen}}
\newcommand{\cmdhash}{\operatorname{\mathcal{H}}}
\newcommand{\cmdsubmit}{\operatorname{submit}}
\newcommand{\cmdsign}{\operatorname{sign}}
\newcommand{\cmdMerkleRoot}{\operatorname{MerkleRoot}} %
\newcommand{\cmdAcceptif}{\text{Accept if }} %
\newcommand{\cmdDetermineRevealorder}{\text{Compute reveal order metric}}

\definecolor{lightgray}{gray}{0.9} %

\newlength{\colwidth}
\setlength{\colwidth}{0.3\textwidth} %
\addtolength{\colwidth}{-2\tabcolsep} %

\begin{figure*}[htbp]
\centering
\begin{tabular}{@{} >{\centering\bfseries}p{\colwidth} @{\hspace{2\tabcolsep}} >{\centering}p{2em} @{} >{\centering\bfseries}p{\colwidth} @{\hspace{2\tabcolsep}} >{\centering}p{2em} @{} >{\centering\bfseries}p{\colwidth} @{}}
Commit & $\longrightarrow$ & Reveal-1 & $\longrightarrow$ & Reveal-2 \\
\end{tabular}

\vspace{-2ex} %

\setlength{\arrayrulewidth}{0.6pt} %
\setlength\dashlinedash{1pt}      %
\setlength\dashlinegap{2pt}       %
\begin{tabular}{ p{\colwidth} : p{\colwidth} : p{\colwidth} } %
\hline
$\begin{aligned}[t] %
  \\[1ex] %
  & s_i \leftarrow \cmdGen() \\
  & c_{o,i} \leftarrow \cmdhash(s_i) \\
  & c_{v,i} \leftarrow \cmdhash(c_{o,i}) \\[0ex] %
  & \cmdsubmit(c_{v,i}) \\[5ex] %
  & \textit{On-chain}
\end{aligned}$
&
$\begin{aligned}[t]
  \\[1ex] %
  & \cmdsubmit(c_{o,i}) \\[2ex]
  & \cmdAcceptif c_{v,i} = \cmdhash(c_{o,i}) \\[2ex]
  & \cmdDetermineRevealorder \\
  & \Omega_v = \cmdhash(c_{o,1} \| \dots \| c_{o,n}) \\
  & d_i = \cmdhash(\Omega_v \| c_{v,i}) %
\end{aligned}$
&
$\begin{aligned}[t]
  \\[1ex] %
  & \cmdsubmit(s_i) \\[2ex]
  & \cmdAcceptif c_{o,i} = \cmdhash(s_i) \\
  & \quad \text{and } d_{i-1} > d_i \\[3ex] %
  & \Omega_o = \cmdhash(s_1 \| \dots \| s_n) \\[3ex] %
\end{aligned}$
\\ \hline
\end{tabular}

\vspace{0.5ex} %
\begin{tabular}{@{}p{\dimexpr\colwidth+\tabcolsep}@{} l @{\hspace{\dimexpr\colwidth+2\tabcolsep}} l @{\hspace{\dimexpr\colwidth+2\tabcolsep}} l @{}}
 & T1 & T2 & T3
\end{tabular}

\caption{On-chain Protocol Description. $\cmdhash$ indicates a hash function.}
\label{fig:onchain-protocol}
\end{figure*}

\begin{figure*}[tbp]
\centering
\begin{tabular}{@{} >{\centering\bfseries}p{\colwidth} @{\hspace{2\tabcolsep}} >{\centering}p{2em} @{} >{\centering\bfseries}p{\colwidth} @{\hspace{2\tabcolsep}} >{\centering}p{2em} @{} >{\centering\bfseries}p{\colwidth} @{}}
Commit & $\longrightarrow$ & Reveal-1 & $\longrightarrow$ & Reveal-2 \\
\end{tabular}

\vspace{-2ex} %

\setlength{\arrayrulewidth}{0.6pt} %
\setlength\dashlinedash{1pt}      %
\setlength\dashlinegap{2pt}       %
\begin{tabular}{@{} p{\colwidth} : p{\colwidth} : p{\colwidth} @{}} %
\hline
\cellcolor{lightgray}%
$\begin{aligned}[t]
  \\[1ex] %
  & s_i \leftarrow \cmdGen(), c_{o,i} \leftarrow \cmdhash(s_i) \\
  & c_{v,i} \leftarrow \cmdhash(c_{o,i}) \\
  & \cmdsubmit(c_{v,i}, \cmdsign_i) \\
  & \text{and accept if verified} \\
  & M_v = \mathcal{M}(c_{v,1}, \dots, c_{v,n}) \\[2ex] %
  & \textit{Off-chain}
\end{aligned}$
& %
\cellcolor{lightgray}%
$\begin{aligned}[t]
  \\[1ex] %
  & \cmdsubmit(c_{o,i}) \\[1ex]
  & \cmdAcceptif c_{v,i} = \cmdhash(c_{o,i}) \\[1ex]
  & \cmdDetermineRevealorder \\
  & \Omega_v = \cmdhash(c_{o,1} \| \dots \| c_{o,n}) \\
  & d_i = \cmdhash(\Omega_v \| c_{v,i}) 
\end{aligned}$
& %
\cellcolor{lightgray}%
$\begin{aligned}[t]
  \\[1ex] %
  & \cmdsubmit(s_i) \\[2ex]
  & \cmdAcceptif c_{o,i} = \cmdhash(s_i) \\
  & \quad \text{and } d_{i-1} > d_i 
\end{aligned}$
\\ \hline
$\begin{aligned}[t]
   & \textit{}
   \\[1ex] %
   & \cmdsubmit(M_v) \\[7ex] %
   & \textit{On-chain}
\end{aligned}$
& %
$\begin{aligned}[t]
   \\[8ex] %
   \ %
\end{aligned}$
& %
$\begin{aligned}[t]
  & \\[1ex]
  & \cmdsubmit(s_1, \cmdsign_1(c_{v,1}), \dots) \\
  & \cmdAcceptif \text{verified and } \\ 
  & M_v = \mathcal{M}(\cmdhash(\cmdhash(s_1)),  \dots) \\
  & \Omega_o = \cmdhash(s_1 \| \dots \| s_n)  \\ 
\end{aligned}$
\\ \hline
\end{tabular}

\vspace{0.5ex} %
\begin{tabular}{@{}p{\dimexpr\colwidth+\tabcolsep}@{} l @{\hspace{\dimexpr\colwidth+2\tabcolsep}} l @{\hspace{\dimexpr\colwidth+2\tabcolsep}} l @{}}
 & T1 & T2 & T3
\end{tabular}

\caption{Off-chain Leveraged Protocol Description. $\cmdhash$ indicates a hash function and $\mathcal{M}$ indicates a Merkle root derivation function.}
\label{fig:offchain-protocol}
\end{figure*}

%% file: 5.accountability.tex
\section{Accountability and Fallback Mechanisms}
\label{sec:fallback}

While cryptographic primitives provide the foundation for security, the robustness of a decentralized protocol in a real-world setting is determined by its ability to handle non-cooperative or faulty participants. The Commit-Reveal² protocol is designed with a core philosophy of accountability: while it may not prevent all liveness failures, it ensures that any party responsible for such a failure can be irrefutably identified and penalized. This is achieved through a set of clearly defined fallback mechanisms that transition from an efficient off-chain process to a more deliberate, fully verifiable on-chain protocol upon detection of a fault.

The complete flow of these accountability procedures is illustrated in Fig. \ref{fig:fallback}. The diagram depicts the protocol's primary operational path and the distinct exception-handling branches triggered by either participant or leader failures. This section details the specific failure scenarios and corresponding resolution mechanisms shown in the flowchart, explaining how each step enforces accountability and ensures protocol integrity.

\subsection{Participant Data Withholding}

The first failure scenario, depicted in the upper branch of Fig.~\ref{fig:fallback}, addresses a non-compliant participant who withholds data during the off-chain coordination phase. To prevent such behavior from stalling the protocol, the leader node is empowered to initiate an on-chain dispute resolution process. This escalation is triggered if a participant fails to broadcast their required values (i.e., $C_v$, $C_o$, or $S$) within the designated submission windows. Depending on the phase of the failure, the leader can invoke a specific smart contract function, such as \texttt{requestToSubmitCv()} or \texttt{requestToSubmits()}, to formally compel the malicious operator to submit their data on-chain.

Once this on-chain fallback mechanism is activated, the malicious participant must submit their data within the \texttt{onChainSubmissionPeriod}. The submission is then subjected to rigorous, multi-stage verification to confirm its validity:
\begin{itemize}
    \item \textbf{Commitment Validity:} The hash of the first-layer commitment must match the second-layer commitment ($C_v = \text{hash}(C_o)$).
    \item \textbf{Reveal Authenticity:} The hash of the revealed secret must match the first-layer commitment ($C_o = \text{hash}(S)$).
    \item \textbf{Merkle Proof Verification:} The commitment must be proven to be part of the original set committed in the Merkle root.
    \item \textbf{Signature Requirement:} To prevent false accusations by the leader, any value not already anchored on-chain must be accompanied by a valid EIP-712 signature from the respective operator.
\end{itemize}
If the participant successfully submits the required data and passes all verifications, the protocol continues. However, failure to comply within the specified timeframe results in severe and irreversible penalties. The non-compliant operator’s staked deposit is immediately slashed, and the funds are redistributed among all compliant operators and the leader node (as compensation for gas costs). Concurrently, the operator is deactivated from the participant set. This ensures that any deviation is either corrected on-chain or results in the offender's swift and costly removal, allowing the protocol to maintain forward progress with the reduced operator set. If the number of active participants falls below the minimum threshold (e.g., two), the protocol enters a halt state until more operators join.

\begin{figure*}[htb]
    \centering
    \includegraphics[width=\linewidth]{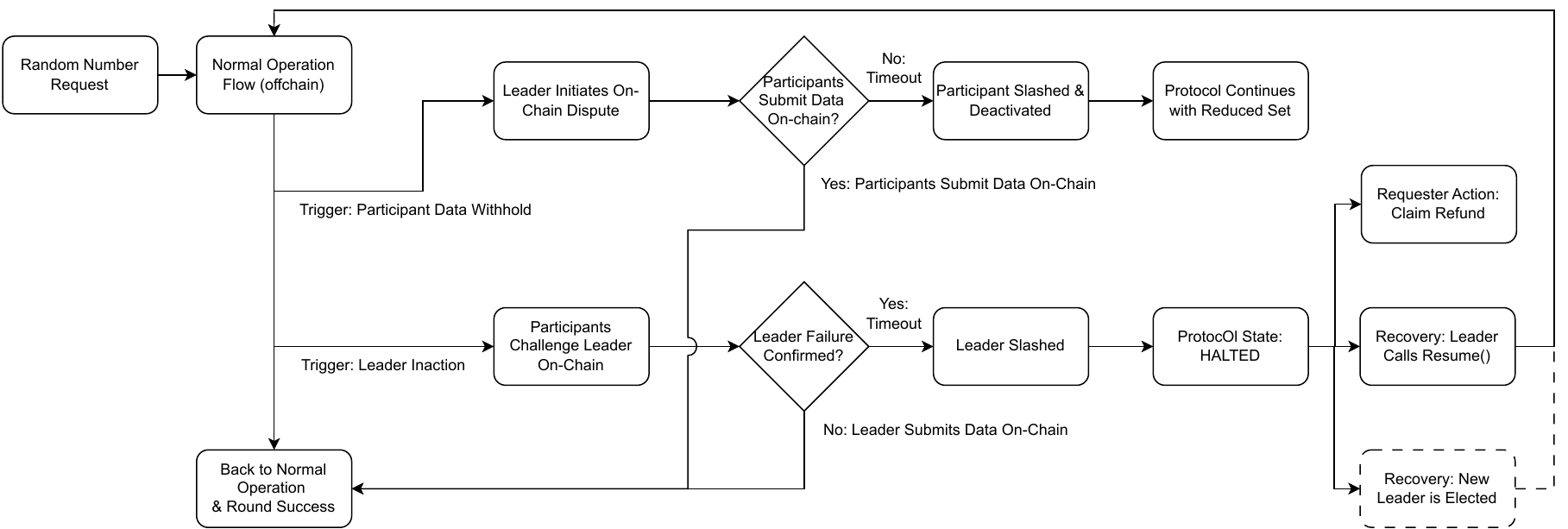}
    \caption{Fallback Mechanism Overview. The diagram illustrates the protocol's response to participant or leader failures. The system operation is triggered by a random number request. It outlines the transition from the normal off-chain operation to on-chain dispute resolution, culminating in either participant slashing, a protocol halt, or successful recovery.}
    \label{fig:fallback}
\end{figure*}

\subsection{Malicious or Failed Leader}
\label{subsec:leader_failure}

The second critical failure scenario, shown in the lower branch of Fig.~\ref{fig:fallback}, involves a malicious or failed leader. The leader's privileged position in aggregating off-chain data and executing on-chain actions makes its reliability paramount. The protocol holds the leader accountable for several time-critical responsibilities:
\begin{enumerate}
    \item \textbf{Merkle Root Submission:} Timely submission of the aggregated $C_v$ commitments' Merkle root.
    \item \textbf{Dispute Initiation:} Initiating on-chain disputes for non-compliant participants.
    \item \textbf{Random Number Generation:} Executing the final random number generation and distribution.
\end{enumerate}

Unlike participant failures detected by the leader, leader failures are detected by any participant through on-chain timeout mechanisms. For instance, if the leader fails to submit the Merkle root within the designated window, any participant can invoke a function to trigger the failure condition. A particularly insidious attack involves the leader submitting an incorrect Merkle root. While the protocol cannot prevent this submission upfront, it ensures eventual detection and punishment, as the incorrect root will cause verification failures during the reveal phase, ultimately leading to a timeout and slashing.

Upon a confirmed leader failure, the protocol imposes immediate and severe penalties. The leader's staked deposit is slashed and redistributed equally among all compliant operators to compensate them for the disruption. Subsequently, the protocol enters a \texttt{HALTED} state to prevent further operations until service can be safely restored. The current recovery model is contingent upon the original leader replenishing their deposit and calling the \texttt{resume()} function. While this design ensures the responsible party bears the cost of recovery, it is tailored for a context with a single designated leader. A more robust approach for a dynamic environment would involve permanently excluding the failed leader and electing a new one to enhance long-term system resilience. The potential mechanisms and trade-offs for such an election process are further explored in our discussion in Section \ref{sec:discussion}.

\subsection{Consumer Protection in a Halted State}

While enforcing operator accountability is crucial, the protocol equally prioritizes protecting consumers who have paid for randomness generation services. The \texttt{HALTED} state, a necessary safeguard against a failed leader, directly impacts these consumers. To address this, the protocol provides a flexible refund mechanism that empowers consumers during service disruptions.

When the protocol is in the \texttt{HALTED} state, a consumer with a pending request faces a choice:
\begin{enumerate}
    \item \textbf{Wait for Resumption:} Trust that the leader will recover and call \texttt{resume()}, allowing their random number request to be fulfilled as originally intended.
    \item \textbf{Claim Refund:} Immediately reclaim their payment by calling the \texttt{refund()} function if they cannot wait or no longer trust the protocol's ability to deliver.
\end{enumerate}
To ensure secure and fair processing, the \texttt{refund()} function is protected by several checks. It verifies that the request belongs to the caller, that the round has not already been processed, and that a refund has not already been issued for that round. Once a refund is claimed, the request is permanently canceled and will be skipped during protocol resumption, preventing any possibility of double-spending. This balanced ecosystem, which separates operator penalties (slashed funds distributed to participants) from consumer restitution (full refunds to consumers), ensures that consumers retain full control over their funds and are never forced to bear the financial consequences of operator failures.

%% file: 6.security_analysis.tex
\section{Security Analysis} \label{sec:security-analysis}

In this section, we prove the security of the proposed $Commit-Reveal^2$ ($CR^2$) based on the definitions in Section \ref{sec:threat model}. Furthermore, we show that $CR^2$ is resistant to replay attacks.

\begin{theorem} \label{thm:unpredictability}
$CR^2$ is unpredictable if the applied hash function $keccak$ behaves as a random oracle.
\end{theorem}

Since our design builds on the standard Commit-Reveal paradigm, the unpredictability of $CR^2$ is established by reduction to the security of Commit-Reveal. The following theorem formalizes this result, with its proof provided in Appendix A.%

	\begin{theorem}
		If the basic Commit-Reveal protocol is unpredictable, then the Commit-Reveal$^2$ protocol is also unpredictable.
	\end{theorem}

	\begin{proof}
		We assume for contradiction that there exists a PPT adversary $\mathcal{A}_{\text{CR}^2}$ who can break the unpredictability of the CR$^2$ protocol with a non-negligible advantage $\epsilon$. This means $\mathcal{A}_{\text{CR}^2}$ can predict the final randomness $\Omega_o = \mathsf{Hash}(s_1 \| \dots \| s_n)$ with probability $\text{Pr}[\mathcal{A}_{\text{CR}^2} \text{ wins}] = \epsilon$, where $\epsilon$ is non-negligible.
		
		We will construct a simulator $\mathcal{S}$ that uses $\mathcal{A}_{\text{CR}^2}$ to break the unpredictability of the basic CR protocol, leading to a contradiction. The simulator $\mathcal{S}$ plays a dual role: it acts as the adversary against a challenger for the basic CR protocol ($\mathcal{C}_{\text{CR}}$) and as the challenger for the CR$^2$ adversary ($\mathcal{A}_{\text{CR}^2}$).
		
		\subsubsection*{Simulator Construction}
		
		\begin{enumerate}
			\item \textbf{Commit Phase of the CR$^2$ Game:} The simulator $\mathcal{S}$ begins its interaction with the CR$^2$ adversary, $\mathcal{A}_{\text{CR}^2}$.
			\begin{itemize}
				\item $\mathcal{A}_{\text{CR}^2}$ chooses $n-1$ secrets, $r_i$, and provides $\mathcal{S}$ with the second-layer commitments: $cv_i = \mathsf{Hash}(\mathsf{Hash}(r_i))$ for $i=1, \dots, n-1$4.
				\item Simultaneously, the simulator $\mathcal{S}$ is running a basic CR game with its challenger, ($\mathcal{C}_{\text{CR}}$). $\mathcal{C}_{\text{CR}}$ provides $\mathcal{S}$ with a commitment from an honest party, $h = \mathsf{Hash}(s)$, where $s$ is the honest secret.
				\item The simulator $\mathcal{S}$ uses this commitment to complete the CR$^2$ commit phase, providing its own commitment to $\mathcal{A}_{\text{CR}^2}$ for the honest party: $cv_h = \mathsf{Hash}(h) = \mathsf{Hash}(\mathsf{Hash}(s))$.
			\end{itemize}
			
			\item \textbf{Reveal-1 Phase of the CR$^2$ Game:} This phase is simulated to allow $\mathcal{A}_{\text{CR}^2}$ to proceed.
			\begin{itemize}
				\item $\mathcal{A}_{\text{CR}^2}$ reveals the first-layer commitments to the simulator: $co_i = \mathsf{Hash}(r_i)$ for $i=1, \dots, n-1$.
				\item The simulator $\mathcal{S}$ reveals its first-layer commitment for the honest party: $co_h = h$.
				\item The simulator then computes the reveal order for the CR$^2$ game, which is a deterministic function of these first-layer commitments.
			\end{itemize}
			
			\item \textbf{Adversary's Prediction and Final Reveal:}
			\begin{itemize}
				\item The simulator's goal is to predict the final output of the basic CR game: $R = \mathsf{Hash}(r_1 \| \dots \| r_{n-1} \| s)$.
				\item $\mathcal{A}_{\text{CR}^2}$ is trying to predict the final output of the CR$^2$ game: $\Omega_o = \mathsf{Hash}(r_1 \| \dots \| r_{n-1} \| s)$.
				\item At the appropriate time, before the final secrets are revealed, $\mathcal{A}_{\text{CR}^2}$ outputs its prediction, $R'$, for $\Omega_o$.
				\item Since the final outputs of both protocols are defined identically, $\mathcal{S}$ can simply forward $\mathcal{A}_{\text{CR}^2}$'s prediction $R'$ to its challenger, $\mathcal{C}_{\text{CR}}$, as its prediction for the basic CR game's output $R$.
			\end{itemize}
		\end{enumerate}
		
		\subsubsection*{Analysis}
		The simulator is successful in breaking the unpredictability of the basic CR protocol if its prediction $R'$ matches the actual output $R$. The information available to the simulator allows it to perfectly simulate the CR$^2$ game for $\mathcal{A}_{\text{CR}^2}$. The final outputs of both games are identical.
		
		The probability that $\mathcal{A}_{\text{CR}^2}$ correctly predicts the CR$^2$ output is $\epsilon$, which is non-negligible by assumption. The probability that $\mathcal{S}$ correctly predicts the basic CR output is therefore also $\epsilon$.
		
		Thus, the simulator $\mathcal{S}$ breaks the unpredictability of the basic Commit-Reveal protocol with a non-negligible advantage. This contradicts our initial premise that the basic CR protocol is unpredictable. Therefore, our initial assumption that an adversary $\mathcal{A}_{\text{CR}^2}$ exists must be false. The Commit-Reveal$^2$ protocol is unpredictable.
	\end{proof}
    	We now prove that the Commit-Reveal$^2$ protocol is bias-resistant if the cryptographic hash function $\mathsf{Hash}$ behaves as a random oracle.
	
	\begin{theorem}
		The Commit-Reveal$^2$ protocol is bias-resistant if the applied hash function acts as a random oracle.
	\end{theorem}
	
	\begin{proof}
		
		We assume for contradiction that a PPT adversary $\mathcal{A}$ can break the bias-resistance of the Commit-Reveal$^2$ protocol with a non-negligible advantage. This means $\mathcal{A}$ can predict a specific bit of the final randomness with a success probability of greater than $\frac{1}{2} + \mathrm{negl}(\lambda)$. We will define a sequence of games to show that no PPT adversary can achieve this.
		
		We assume there are $n$ total participants and the adversary $\mathcal{A}$ controls $t$ of them, while $n-t$ participants are honest.
		
		\begin{itemize}
			\item $\mathsf{Game}_0$ (The Real Protocol)
			\begin{itemize}
				\item \textbf{Protocol Execution}: The Commit-Reveal$^2$ protocol runs as described. The adversary $\mathcal{A}$ controls $t$ parties and observes the on-chain commitments and off-chain revelations of the honest parties.
				\item \textbf{Adversary's Task}: Before the final reveal phase, the adversary $\mathcal{A}$ chooses a bit position $k$ and outputs a prediction for the $k^{th}$ bit of the final randomness, $\Omega_o[k]$.
				\item \textbf{Win Condition}: The adversary wins if its prediction for $\Omega_o[k]$ is correct.
				\item \textbf{Advantage}: The adversary's advantage in this game is defined as $\text{Adv}_0 = \text{Pr}[\mathcal{A} \text{ correctly predicts } \Omega_o[k]]$. For bias-resistance, we assume $\text{Adv}_0 > 1/2 + \text{negl}(\lambda)$.
			\end{itemize}
			
			\item $\mathsf{Game}_1$ (Ideal Hashing of Honest Secrets)
			\begin{itemize}
				\item \textbf{Change}: This game is identical to $\mathsf{Game}_0$, except that the secrets $s_j$ of the $n-t$ honest parties are replaced with uniformly random, independent strings.
				\item \textbf{Protocol Execution}: The simulator runs the protocol, providing the adversary with the commitments that correspond to the randomly chosen honest secrets. The simulator uses a random oracle to "program" the hash function, so the commitments and revelations for the honest parties are consistent with the random secrets.
				\item \textbf{Adversary's Task \& Win Condition}: These are the same as in $\mathsf{Game}_0$.
				\item \textbf{Indistinguishability}: The adversary's view in $\mathsf{Game}_0$ and $\mathsf{Game}_1$ is computationally indistinguishable because the $\mathsf{Hash}$ function behaves as a random oracle. An adversary cannot distinguish a hash of a high-entropy secret from a hash of a uniformly random string. Therefore, the advantage difference is negligible: $\text{Adv}_0 \approx \text{Adv}_1$.
			\end{itemize}
			
			\item $\mathsf{Game}_2$ (Random Final Output)
			\begin{itemize}
				\item \textbf{Change}: This game is identical to $\mathsf{Game}_1$, except for the final output. Instead of computing $\Omega_o = \mathsf{Hash}(s_1 \| \dots \| s_n)$, the simulator chooses a final output $R^*$ as a uniformly random string from $\{0,1\}^m$ and "programs" the random oracle to return $R^*$ when queried with the concatenated secrets $s_1 \| \dots \| s_n$.
				\item \textbf{Adversary's Task}: The adversary's goal is to predict the $k^{th}$ bit of $R^*$.
				\item \textbf{Win Condition}: The adversary wins if its prediction for $R^*[k]$ is correct.
				\item \textbf{Indistinguishability}: The adversary's view in $\mathsf{Game}_1$ and $\mathsf{Game}_2$ is also computationally indistinguishable. If the hash function is a random oracle, its output is indistinguishable from a truly random string. Thus, the advantage difference between $\mathsf{Game}_1$ and $\mathsf{Game}_2$ is negligible: $\text{Adv}_1 \approx \text{Adv}_2$.
			\end{itemize}
		\end{itemize}
		
		The adversary's advantage in $\mathsf{Game}_2$ is the probability of guessing a single bit of a truly random string. This probability is exactly $1/2$, as each bit of a random string is equally likely to be 0 or 1.
		
		Since $\text{Adv}_0 \approx \text{Adv}_1$ and $\text{Adv}_1 \approx \text{Adv}_2$, it follows that $\text{Adv}_0 \approx \text{Adv}_2$. This means the adversary's advantage in the real protocol ($\mathsf{Game}_0$) is negligibly close to $1/2$. This contradicts our initial assumption that $\text{Adv}_0 > 1/2 + \text{negl}(\lambda)$, thereby proving that the Commit-Reveal$^2$ protocol is bias-resistant.
		
	\end{proof}

\begin{theorem}[Replay Resistance]
In the Commit--Reveal$^2$ protocol (Section~\ref{sec:protocol}), 
each participant $i$ binds its commitment $c_{v,i}$ to a specific execution context by either:
\begin{itemize}
  \item considering the \emph{fully on-chain path}, each participant submits $c_{v,i}$ individually; or
  \item considering the \emph{hybrid path}, producing a digital signature 
  $\sigma_i \leftarrow \mathsf{Sign}_{\mathsf{sk}_i}(M_i)$ on the structured message $$\small{M_i = (
\texttt{chainId},\; \texttt{verContract}, \texttt{round},\; \texttt{attemptId},\; c_{v,i})}$$
 with $c_{v,i}$ anchored in the on-chain Merkle root of commitments.
\end{itemize}

We have the following four assumptions:
\begin{enumerate}
  \item[\textbf{A1}] The signature scheme is existentially unforgeable under chosen-message attack (EUF-CMA).
  \item[\textbf{A2}] The tuple $\small{(\texttt{chainId},\texttt{verContract},\texttt{round},\texttt{attemptId})}$
  is encoded injectively inside $M_i$ (domain separation).
  \item[\textbf{A3}] The Merkle hash function is collision- and second-preimage resistant.
  \item[\textbf{A4}] The contract’s $\mathsf{Seen}$ map is persistent and respects ledger finality.
\end{enumerate}
Then for every PPT adversary $\mathcal{A}$,
\[
\Pr[\mathcal{A}_{\mathsf{replay}}]
\;\le\;
\Adv^{\mathrm{euf\mbox{-}cma}}_{\mathsf{Sig}}(\lambda)
\;+\;
\Adv^{\mathrm{coll/2\mbox{-}pre}}_{\mathcal{H}}(\lambda)
\;+\;
\mathsf{negl}(\lambda),
\]
and hence the protocol is replay resistant.
\end{theorem}

\begin{proof}
A \emph{replay} means the contract accepts data authorized for 
context $\small{(\texttt{chainId},\texttt{verContract},\texttt{round},\texttt{attemptId})}$
either (i) in a different context (cross-context) 
or (ii) twice in the same context (intra-context).

\textbf{On-chain path.}
Acceptance records $\mathsf{Seen}[(\mathsf{addr},\texttt{round},\texttt{attemptId})]\!\leftarrow\!\mathsf{true}$.
Cross-context replay requires acceptance under mismatched $(\texttt{round}',\texttt{attemptId}')$, 
which the contract rejects unless state/finality fails (\textbf{A4}). 
Intra-context replay requires re-acceptance of the same tuple, which $\mathsf{Seen}$ forbids. 
Thus replay advantage is negligible.

\textbf{Hybrid path.}
Acceptance requires (a) a valid signature $\sigma_i$ on $M_i$ 
and (b) a valid Merkle inclusion of $c_{v,i}$.

(i) \emph{Cross-context.} Re-using $(c_{v,i},\sigma_i)$ in a different context 
amounts to verifying $\sigma_i$ on $M_i' \neq M_i$. 
By \textbf{A2}, messages are context-unique; success implies an EUF-CMA forgery (\textbf{A1}).  

(ii) \emph{Intra-context.} Reuse within the same $(\texttt{round},\texttt{attemptId})$ 
requires a second acceptance at $(i,\texttt{round},\texttt{attemptId})$, 
which $\mathsf{Seen}$ rejects (\textbf{A4}).  

(iii) \emph{Commitment-set substitution.} 
Accepting $(c_{v,i},\sigma_i)$ under a different Merkle root requires either 
a false inclusion proof or changing the set without changing the root, 
contradicting collision or second-preimage resistance of $\mathcal{H}$ (\textbf{A3}).  

By union bound over (i)--(iii), the replay advantage is at most
\[
\Adv^{\mathrm{euf\mbox{-}cma}}_{\mathsf{Sig}}(\lambda) + 
\Adv^{\mathrm{coll/2\mbox{-}pre}}_{\mathcal{H}}(\lambda) + 
\mathsf{negl}(\lambda).
\]
\end{proof}

In the following sections, we proceed to validate the design in practice: Section \ref{sec:implementation} outlines the implementation, and Section \ref{sec:cost analysis} measures the resulting gas and fallback overhead.

%% file: 7.implementation.tex
\section{Implementation}
\label{sec:implementation}

This section details the Commit-Reveal\(^2\) protocol implementation, contrasting fully on-chain and hybrid models. It explains the chosen strategies and gas cost measurement methods, while the appendix provides additional information on Merkle Root construction, EIP-712 typed data hashing, and a link to the full source code can be accessed in the repository:  
\url{https://github.com/tokamak-network/Commit-Reveal2/tree/full-paper}.

\subsection{Merkle Root Construction}

The Commit-Reveal\(^2\) protocol constructs a Merkle Root as a complete binary tree to enable consistent and verifiable computation. Each leaf corresponds to a participant's \texttt{$c_{v}$} values and is sorted by the order of activated nodes to ensure deterministic results. This organized structure allows the root to be efficiently computed and used for validation in later phases. The detailed implementation is provided in Appendix C.

\subsection{Signature Construction}
\label{subsec:signatureconstruction}

The protocol uses EIP-712~\cite{eip712} for signing and hashing data under Ethereum’s cryptographic standards, employing a domain separator and a defined message structure.

The domain separator encodes protocol details such as the \texttt{chainId} and contract address to bind each signature to the intended blockchain and contract. The message structure includes \texttt{round} and \texttt{trialNum} as nonces, which together provide uniqueness and protect against replay attacks. An EIP-712-compliant hash merges these elements into a digest, which is then signed with ECDSA to produce ((v, r, s)). Further details on hash construction are provided in Appendix D.

\subsection{Hybrid Commit-Reveal² Implementation}

This subsection outlines the hybrid implementation of the Commit-Reveal\(^2\) protocol. The hybrid approach minimizes on-chain gas usage by delegating computationally intensive operations to off-chain systems while maintaining integrity and transparency through on-chain validation and random number generation.

\subsubsection{Commit Phase}

In this phase, a participant submits a Merkle Root representing their hashed commitments (\texttt{$c_{v,i}$}) to the smart contract. The Merkle Root is stored on-chain for later validation. The contract ensures that only activated participants can submit the Merkle Root, verified against their activation order.

\begin{lstlisting}[caption={Interface for the Generate Random Number Function}, label={code:generate-random-interface}]
function generateRandomNumber(
    bytes32[] calldata secrets,
    uint8[] calldata vs,
    bytes32[] calldata rs,
    bytes32[] calldata ss
) external;
\end{lstlisting}

\subsubsection{Reveal-2 Phase}

During this phase, the secrets, along with the related signatures are sent to the blockchain, where the smart contract performs all necessary validations before generating the random number. The interface for this phase is provided in Code~\ref{code:generate-random-interface}.

The implementation validates inputs and ensures fairness through the following steps:

\paragraph{Commitment Verification.}
Each \texttt{$c_{o,i}$} is computed by hashing the secret, and each \texttt{$c_{v,i}$} is derived by hashing the corresponding \texttt{$c_{o,i}$}.

\paragraph{Merkle Root Validation.}
The Merkle Root, submitted in Phase 1, is reconstructed from the \(\{c_{v,1}, \ldots, c_{v,n}\}\) and compared with the stored Merkle Root. This ensures that the revealed secrets correspond to the original commitments made during the commitment phase.

\paragraph{Signature Validation.}
Each participant’s secret is authenticated using EIP-712-compliant signatures. The \texttt{ecrecover} function extracts the signer’s public key from their signature, validating that the recovered address is an activated operator address. To prevent signature malleability, the protocol enforces that the \(s\)-value of the ECDSA signature remains within the lower half of the SECP256k1 curve, using the constant \texttt{0x7FFFFFFFFFFFFFFFFFFFFFFFFFFFFFFF5D576E73\allowbreak57A4501DDFE92F46681B20A0}~\cite{eip2}. This constraint eliminates the risk of multiple valid signatures for a single elliptic curve point.

\paragraph{Random Number Generation.}
After all validations, the revealed \texttt{secrets} are combined in the activation order and hashed to compute the final random number. This ensures fairness, unpredictability, and adherence to the protocol's cryptographic guarantees.

\subsection{Fully On-Chain Commit-Reveal\(^2\) Implementation}

The fully on-chain protocol operates in three distinct phases:

\subsubsection{Commit Phase} Participants submit their \texttt{$c_{v,i}$} to the smart contract, which are stored on-chain along with their submission order. This order is later used to verify the sequence in which participants reveal their secrets.

\subsubsection{Reveal-1 Phase} Participants disclose their \texttt{$c_{o,i}$}, which are verified against their corresponding \texttt{$c_{v,i}$}. During this phase, the contract pre-allocates state memory to store the \texttt{$c_{o,i}$} array in the order of the original submissions. The first participant incurs slightly higher gas costs due to this initialization.

\subsubsection{Reveal-2 Phase} Participants reveal their \texttt{$s_{i}$} in a sequence determined by a hash-based random value. The first revealer verifies and stores the reveal order, while the final participant combines all verified \texttt{secrets} to compute the random number.

\section{Cost Analysis}
\label{sec:cost analysis}

This section quantifies the protocol’s gas costs. We first report the normal-path measurements, then evaluate the overhead of fallback (abnormal) paths under participant- and leader-withholding. Finally, we assess the practicality and impact of late-reveal griefing.

\subsection{Experimental Environment}

The experiments were conducted using Foundry, a robust testing framework for Ethereum development. The configuration included enabling the optimizer with \texttt{optimization\_runs} set to its maximum value (\texttt{4294967295}, equivalent to \(2^{32}-1\)), and activating the Intermediate Representation mode (\texttt{via-ir}). Tests were executed on the Cancun EVM version with the \texttt{0.8.30} compiler, utilizing both Solidity and Yul. Gas usage was measured using \texttt{vm.lastCallGas().gasTotalUsed} from Foundry's test cheats library, with the \texttt{isolate} flag enabled to execute each top-level call in a separate EVM context for precise gas accounting and state tracking. 

\noindent\textit{Data availability.} The measurements underlying the figures and tables in this section are available in the public repository; see Appendix B.

\subsection{Normal path cost}

To assess the gas efficiency of the Commit-Reveal\(^2\) implementation, a comparative analysis was performed between the fully on-chain approach and the hybrid off-chain leveraged model. We measured the total gas consumed for a single round of the protocol to provide a clear comparison between the two models.

\begin{figure}[tb]
    \centering
    \includegraphics[width=0.93\linewidth]{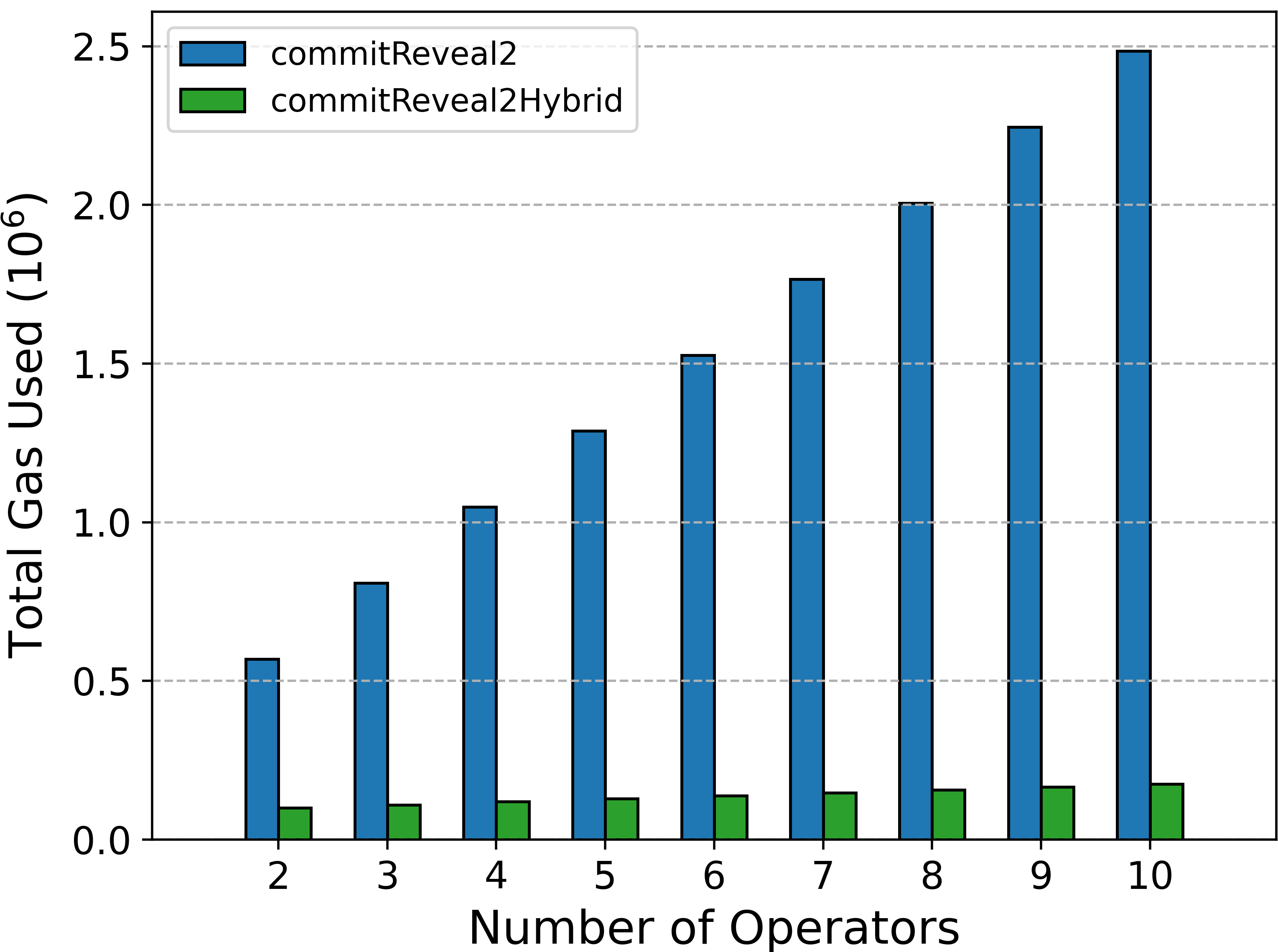}
    \caption{Gas cost comparison between fully on-chain and hybrid models.}
    \label{fig:gas-comparison}
\end{figure}

Figure~\ref{fig:gas-comparison} highlights the dramatic gas cost differences between the two approaches. The fully on-chain implementation starts at approximately 569,412 gas for two operators and exceeds 2,484,566 gas for ten operators, driven by the increasing number of transactions and on-chain storage operations. In contrast, the hybrid model reduces gas costs by over 80\%, starting at around 100,732 gas for two operators and scaling to 175,569 gas for ten operators. This efficiency is achieved by minimizing on-chain storage operations while maintaining transparency and verifiability through cryptographic proofs.

Using recent Ethereum network data (June--August 2025, post-Pectra upgrade, with average gas price of 2.85 Gwei and average Ethereum price of 3,133 USD), the transaction fees for the hybrid model can be expressed in dollars. For example, for three operators, the total gas used is 110,065, resulting in a fee of approximately 0.98 USD. The hybrid model's efficient gas consumption of only 110,065 units for three operators highlights its cost-effectiveness, making it a practical solution for scalable blockchain applications.

\begin{figure}[!t]
    \centering
    \includegraphics[width=\linewidth]{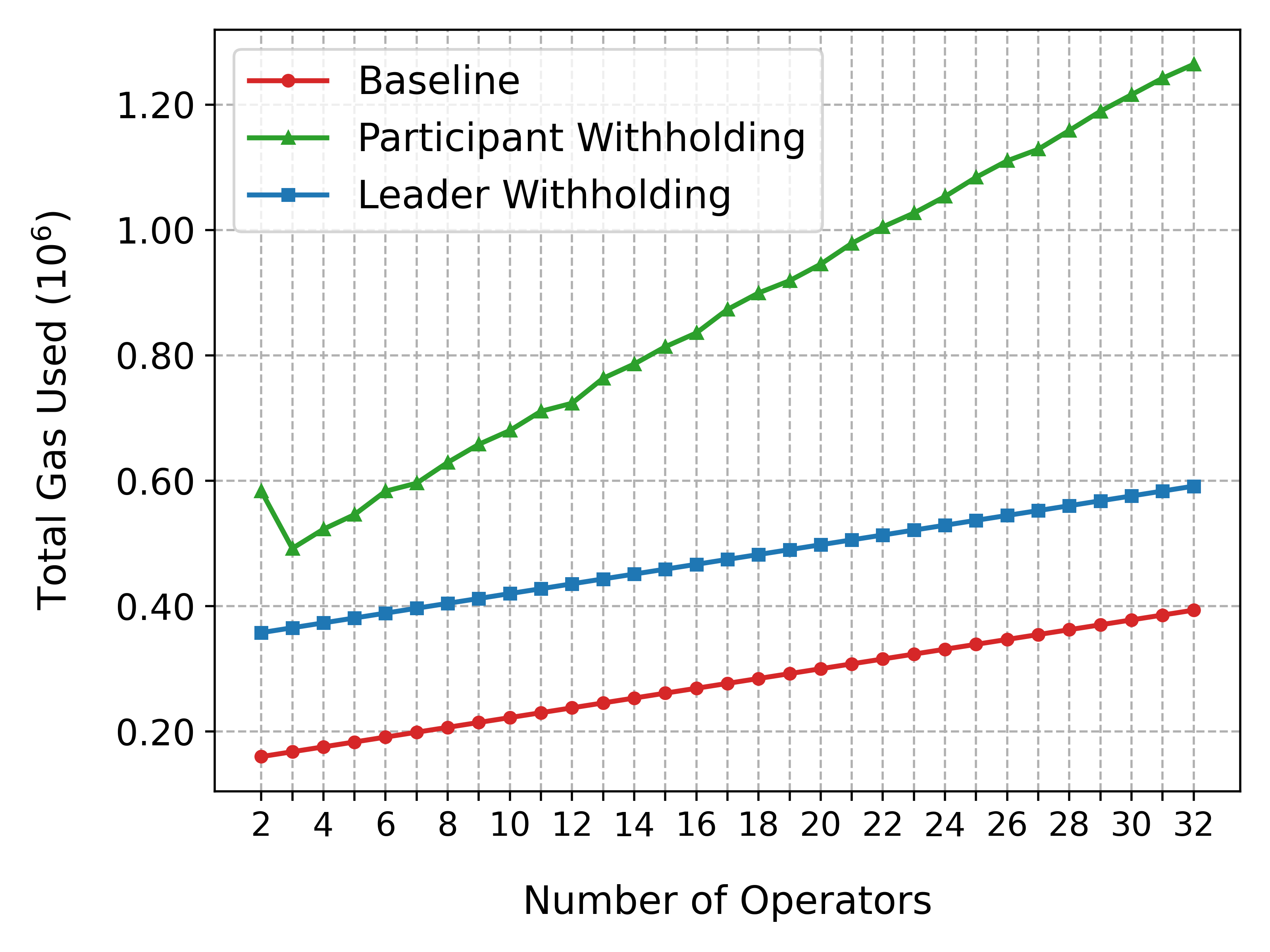}
    \caption{Abnormal‑path (fallback enabled) gas usage by the number of operators. The precise routes used in measurement for The baseline, participant‑withholding, and leader‑withholding paths are listed in Appendix E.}
    \label{fig:abnormal-full}
\end{figure}

\begin{table*}[htb]
\centering
\caption{Representative abnormal‑path costs (gas). $\Delta$ is the increase over the baseline; ratios are scenario/baseline.}
\label{tab:abnormal-cost}
\small
\begin{tabular}{c|c|c|c|c|c|c|c}
\toprule
$n$ & Baseline &
\shortstack{\textbf{Participant}\\\textbf{Withholding}} &
\shortstack{\textbf{Leader}\\\textbf{Withholding}} &
\shortstack{$\boldsymbol{\Delta}$\\\textbf{(PW)}} &
\shortstack{$\boldsymbol{\Delta}$\\\textbf{(LW)}} &
\shortstack{\textbf{PW}/\\\textbf{B}} &
\shortstack{\textbf{LW}/\\\textbf{B}} \\
\midrule
2  & 159{,}969 & 583{,}609 & 357{,}751 & 423{,}640 & 197{,}782 & 3.65$\times$ & 2.24$\times$ \\
10 & 222{,}251 & 680{,}218 & 420{,}052 & 457{,}967 & 197{,}801 & 3.06$\times$ & 1.89$\times$ \\
20 & 300{,}159 & 945{,}530 & 497{,}933 & 645{,}371 & 197{,}774 & 3.15$\times$ & 1.66$\times$ \\
32 & 393{,}621 & 1{,}264{,}275 & 591{,}403 & 870{,}654 & 197{,}782 & 3.21$\times$ & 1.50$\times$ \\
\bottomrule
\end{tabular}
\end{table*}

\subsection{Abnormal Path Cost}
\label{subsec:abnormal-cost}
When the protocol leaves its normal off‑chain path and enters the fallback mechanism, a fixed amount of verification must be performed on chain even if no one ultimately misbehaves. To make that overhead visible, we compare three paths that capture the two principal failure classes introduced in Section~\ref{sec:fallback}: a baseline where everyone eventually cooperates (with the fallback logic enabled), a case where a participant (not a leader) withholds, and a case where the leader fails once and then service resumes. Because each failure case takes multiple concrete routes depending on timing and what is already recorded on chain, our measurements use one representative route per case. The exact call sequences used are listed in Appendix E.

In the normal path cost comparison, the baseline does not include fallback checks, whereas here they are enabled for all three curves. The “baseline” used here is not the hybrid normal path shown earlier in Figure~\ref{fig:gas-comparison}. There, the Commit‑Reveal\(^2\) hybrid flow was measured \emph{without} fallback. Here, the fallback machinery is enabled even when everyone cooperates: the contract still initializes timers/flags and maintains minimal state so it could invoke fallback if needed, and that fixed bookkeeping costs gas each round. As a result, the baseline rises from $\approx$175{,}569 to 222{,}251 gas at $n{=}10$ (+$\approx$46.7k), and from $\approx$100{,}732 to 159{,}969 at $n{=}2$ (+$\approx$59.2k).

Figure~\ref{fig:abnormal-full} plots gas usage against the number of operators, and Table~\ref{tab:abnormal-cost} lists representative points. The Leader‑withholding (LW) curve closely follows the baseline with an almost constant premium of about 198k gas across the measured range: one failed attempt, a resume, and a second try add fixed work but do not change how costs scale with the operator set. By contrast, the Participant‑Withholding (PW) curve grows more quickly. When a participant forces the fallback, the chain performs round‑level validation once per operator: it applies EIP‑712 signature checks when a commitment has not yet appeared on chain, verifies the globally agreed reveal order that ensures deterministic finalization (including the last‑revealer reward), and confirms that the commitments as a set are consistent with the previously committed Merkle root. Because these steps are applied across the operator set, they add an almost constant amount of gas per operator on top of the baseline, which explains the near‑linear slope observed in the plot. The two‑operator case shows an extra jump because the set briefly dips below the minimum threshold and a short recovery is required before progressing.

To summarize, our analysis shows that although the fallback mechanism preserves correct randomness, it imposes additional costs and a bounded delay on other participants. Because this shifts burden while preserving the protocol’s liveness and correctness guarantees, we treat it as griefing and examine its practicality and impact in the next subsection.

\subsection{Practicality of Griefing Attack}
\label{subsec:griefing}

A participant can deliberately stay silent off chain so that the protocol must fall back to its on‑chain procedure, and then still submit within the allowed on‑chain period. Because the submission does arrive, this behavior does not trigger slashing. Its impact is twofold: it shifts extra work to others (mainly the leader, who must open the on‑chain request and verify the round on chain) and it slows the round by a bounded amount determined by the protocol’s period parameters.

The measurements show that the leader’s on‑chain request scales almost perfectly linearly with the size of the operator set. A least‑squares fit over \(n\in[3,32]\) is
\begin{equation}
  G(n) \;=\; \mathbf{18{,}563}\,n \;+\; \mathbf{108{,}808}\ \text{gas},
  \label{eq:reqs-fit}
\end{equation}
with a coefficient of determination \(R^2>0.99999\). The slope comes from work that is applied once to each operator in the round: when a commitment has not appeared on chain it must be accompanied by an EIP‑712 signature and checked against the activated operator address; the reveal order for the round is validated to ensure deterministic finalization (and the last‑revealer reward); and the set of commitments is checked against the previously committed Merkle root. The intercept reflects fixed round overhead that is present regardless of \(n\).

By contrast, the griefer issues a single on‑chain submission. In our implementation, once the protocol falls back on chain, finalization—computing the random number and invoking the consumer’s callback—runs together with the \emph{last} secret submission in the fixed reveal order. If the griefer is that last submitter, this finalization work is attached to his transaction. The finalization term itself grows only mildly with the operator set (about 2.18k gas per operator; e.g., 125{,}661 at $n{=}2$, 130{,}021 at $n{=}4$, 136{,}561 at $n{=}7$), whereas the leader’s on‑chain request remains the dominant component with \(\approx\)18.6k gas per operator (Eq.~\ref{eq:reqs-fit}). Table~\ref{tab:griefing-s} lists representative values and shows how the asymmetry widens as \(n\) increases.

\begin{table}[htb]
\centering
\caption{Late on‑chain reveal (fallback) case: leader vs.\ attacker gas at representative \(n\).}
\label{tab:griefing-s}
\small
\setlength{\tabcolsep}{4pt}
\renewcommand{\arraystretch}{1.08}
\begin{tabular}{r|c|c|c}
\toprule
$n$ & \shortstack{Leader\\request (gas)} & \shortstack{Attacker\\(one submit)} & \shortstack{Leader/\\Attacker $\ge$} \\
\midrule
3  & 164{,}517 & 46{,}821 & 3.51$\times$ \\
10 & 294{,}413 & 46{,}821 & 6.29$\times$ \\
20 & 480{,}081 & 46{,}821 & 10.25$\times$ \\
32 & 702{,}839 & 46{,}821 & 15.01$\times$ \\
\bottomrule
\end{tabular}
\end{table}

Summing up, the behavior matches a classic griefing pattern: no direct attacker gain, cost-shifting to others, and bounded delay with the same final outcome. Accountability and eventual liveness are preserved, and the responsible party is identifiable on-chain. A simple extension could split the on-chain request cost between the leader and the late submitter (e.g., 50–50) without affecting asymptotic scaling. We leave this refinement to future work.

%% file: 8.discussion.tex
\section{Discussions} \label{sec:discussion}

\subsection{Enhanced Considerations for BLS Integration in Smart Contracts}
Our work demonstrates significant gas savings through off-chain computations; however, as the number of participants grows, the linear complexity of individual signature verifications can escalate costs. To address this, the Boneh-Lynn-Shacham (BLS) signature scheme \cite{boneh2001short} offers a means to aggregate and verify multiple signatures in a single operation. Prior research suggests that BLS verification becomes cost-effective when verifying at least 38 signatures \cite{liang2023bls}. While this provides a useful guideline, practical adoption should consider the specific scale and operational conditions of the target protocol—such as typical participant counts, expected verification frequencies, and on-chain resource constraints. Additionally, the pairing-based computations required for BLS incur non-negligible costs. Thus, before integrating BLS, designers should weigh these overheads against the scalability benefits and examine potential optimizations (e.g., efficient libraries or hardware accelerations) that might mitigate the computational load.

\subsection{Applicability to the Commit-Reveal-Recover Scheme}
While the Commit-Reveal mechanism can generate randomness even under partial withholding, its liveness issues remain structurally challenging. A well-known solution is the Commit-Reveal-Recover scheme, which employs VDF-based timed commitments to recover randomness as long as a minimum threshold of honest participants exists. However, practical concerns—limited availability of VDF-specific ASICs and potential parallelization attacks \cite{leurent2023analysis, barakbayeva2024srng} that weaken VDF security—currently hinder its widespread use. Despite these hurdles, our protocol’s straightforward hash-based commitments can be replaced by timed commitments, integrating seamlessly into a Commit-Reveal-Recover framework. This integration would allow the protocol to maintain robustness against withholding, mitigate delays, and deter last revealer attacks.

\subsection{Decentralized Leader Election for Enhanced Liveness}
Upon a confirmed leader failure, the protocol slashes the leader and enters a \textsc{HALTED} state. In the current recovery model, progress depends on the original leader replenishing the deposit and calling \texttt{resume()}, which creates a single point of failure for liveness if the leader remains unresponsive.
To address this, we propose an in-protocol, low-overhead mechanism that restores liveness without external governance while preserving accountability and slashing as a future work. The design remains compatible with existing Commit-Reveal\textsuperscript{2} verification paths (e.g., EIP-712 signatures and Merkle-based checks) and uses a simple, deterministic selection rule.

After the system enters \textsc{HALTED} and a fixed election delay elapses, any compliant operator may start an \emph{election round}. The remaining active operators execute a lightweight, on-chain \emph{commit--reveal mini-round} to derive election randomness \(R_{\text{elec}}\). A new leader is then chosen as a deterministic function of \(R_{\text{elec}}\) and the candidate set (for example, \(\arg\min_{i}\ \mathrm{Hash}(R_{\text{elec}}\parallel \textsf{addr}_i)\)); optionally, stake-weighted selection can be applied. Finalization is performed by the elected leader in a single transaction, and round identifiers/timestamps prevent duplicate elections within the same recovery instance. Because this procedure reuses the existing verification logic, both implementation effort and gas overhead remain modest.

If participation falls below a threshold, the election round expires and may be retried with exponential backoff. The mechanism coexists with consumer refunds in the \textsc{HALTED} state and maintains the original accountability rules including slashing. It introduces a one-off, on-demand cost and does not change the asymptotic scaling of the normal or fallback paths. A full parameter study (timeouts, thresholds, weighting) and a detailed gas evaluation should be followed.

%% file: 9.conclusion.tex
\section{Conclusion}
\label{section:conclusion}

We introduced \emph{Commit--Reveal\textsuperscript{2}}, a layered commit--reveal protocol that cryptographically randomizes the reveal order and mitigates the strategic leverage of the last revealer. The design is paired with an on-chain accountability framework (slashing, dispute handling, and fallbacks) that targets real-world deployments while preserving auditability. We formalized unpredictability and bias-resistance under standard assumptions and released a hybrid implementation that keeps verification on chain while offloading routine work off chain, yielding more than 80\% gas savings over a fully on-chain baseline for small to moderate operator sets.

Empirically, the normal path maintains a low and stable gas profile, and enabling fallbacks raises the baseline modestly due to the bookkeeping required for accountability. Under participant withholding, overhead grows roughly with the operator set size. Occasional late on-chain submissions can shift cost and add bounded delay without changing the outcome; even in such cases, the protocol preserves accountability and conditional liveness, and users are protected by a refund mechanism while the system is paused.

Potential directions for future work include refining an in-protocol leader replacement to reduce reliance on the original leader for recovery, evaluating when on-chain aggregate BLS verification becomes cost-effective at larger scales, and exploring interoperability with variants that incorporate timed commitments. These directions may improve conditional liveness and cost efficiency without changing how the protocol scales as the system grows.

%% file: appendix.tex
\appendix

\subsection{Unpredictability Proof for Commit-Reveal Protocol}\label{appendix:cr_proof}
	
	\begin{theorem}
		The Commit-Reveal protocol is unpredictable if the cryptographic hash function behaves as a random oracle.
	\end{theorem}
	
	\begin{proof}
		We prove unpredictability using the game-hopping technique in the random oracle model. We assume, for contradiction, that a probabilistic polynomial-time (PPT) adversary $\mathcal A$ can predict the final randomness $R$ of the Commit-Reveal protocol with a non-negligible advantage. We define a sequence of games, where the first game is the real protocol and the last game is a purely random process. We will show that no PPT adversary can distinguish between any two adjacent games, which implies that the adversary's initial non-negligible advantage is wrong.
		
		Let $n$ be the number of parties, and we assume at least one party is honest. The adversary $\mathcal A$ controls at most $n-1$ parties, and one party, $\mathcal P_h$, is honest.
		
		\begin{itemize}
			\item \textbf{Game 0 (The Real Protocol)} 
			\begin{itemize}
				\item \textbf{Commit Phase:} For each party $\mathcal P_i$, ($i=1,\dots,n-1$), the adversary $\mathcal A$ chooses a secret (random or not) $r_i$ and computes $c_i=\mathsf{Hash}(r_i)$. For the honest party $\mathcal P_h$, a secret $r_h$ is chosen uniformly at random. The commitment $c_h=\mathsf{Hash}(r_h)$ is computed. All commitments are broadcast.
				\item \textbf{Adversary's Guess:} Before the reveal phase, the adversary $\mathcal A$ outputs a guess $R_{\mathcal A}$ for the final randomness.
				\item \textbf{Reveal Phase:} All parties reveal their secrets. The final randomness is computed as $R=\mathsf{Hash}(r_1||\dots||r_{n-1}||r_h)$.
				\item \textbf{Win Condition:} The adversary wins if $R_{\mathcal A}=R$.
				\item \textbf{Advantage:} The adversary's advantage in this game is defined as $\mathsf{Adv}_0=\Pr[R_{\mathcal A}=R]$.
			\end{itemize}
			
			\item \textbf{Game 1 (Random Honest Commitment)} 
			\begin{itemize}
				\item \textbf{Change:} This game is identical to $\mathsf{Game}_0$, except for the honest party's commitment. Instead of computing $c_h=\mathsf{Hash}(r_h)$, the simulator chooses $c_h$ as a uniformly random string from $\{0,1\}^m$.
				\item \textbf{Commit Phase:} The adversary's commitments $c_i$ are computed as in $\mathsf{Game}_0$. For the honest party $\mathcal{P}_h$, a random $r_h$ is chosen, but the commitment $c_h$ is a randomly chosen string.
				\item \textbf{Reveal Phase:} To maintain consistency, the simulator uses a random oracle to "program" the hash function. When the hash oracle is queried with $r_h$, it is made to return the pre-chosen random commitment $c_h$. The verification step $c_h=\mathsf{Hash}(r_h)$ will pass.
				\item \textbf{Adversary's Guess \& Win Condition:} These are the same as in $\mathsf{Game}_0$.
				\item \textbf{Indistinguishability:} The adversary's view in $\mathsf{Game}_0$ and $\mathsf{Game}_1$ is computationally indistinguishable because Hash acts as a random oracle. The only difference is whether $c_h$ is a true hash of a random value or just a random value itself. If an adversary could distinguish between these two scenarios with non-negligible probability, it would break the pseudorandomness property of the hash function. Therefore, the advantage difference is negligible: $\mathsf{Adv}_0\approx\mathsf{Adv}_1$.
			\end{itemize}
			
			\item \textbf{Game 2 (Random Final Output)} 
			\begin{itemize}
				\item \textbf{Change:} This game is identical to $\mathsf{Game}_1$, except for the computation of the final randomness. Instead of computing $R=\mathsf{Hash}(r_1||\dots||r_{n-1}||r_h)$, the simulator chooses $R^*$ as a uniformly random string from $\{0,1\}^m$.
				\item \textbf{Commit and Reveal Phases:} These are the same as in $\mathsf{Game}_1$.
				\item \textbf{Final Output:} The final output is simply $R^*$. The simulator programs the random oracle to return $R^*$ when queried with the concatenated secrets $r_1||\dots||r_{n-1}||r_h$, maintaining consistency.
				\item \textbf{Adversary's Guess \& Win Condition:} The adversary guesses $R_{\mathcal A}$ and wins if $R_{\mathcal A}=R^*$.
				\item \textbf{Indistinguishability:} The adversary's view in $\mathsf{Game}_1$ and $\mathsf{Game}_2$ is also computationally indistinguishable because the final output is either a hash of a random-looking string (from $\mathsf{Game}_1$) or a truly random string (from $\mathsf{Game}_2$). If an adversary could distinguish these, it would again violate the random oracle model. Therefore, the advantage difference is negligible: $\mathsf{Adv}_1\approx\mathsf{Adv}_2$.
			\end{itemize}
		\end{itemize}
		
		The adversary's advantage in $\mathsf{Game}_2$ is the probability of guessing a random string $R^*$ of length $m$. This probability is $2^{-m}$, which is negligible for a sufficiently large $m$.
		
		Since $\mathsf{Adv}_0\approx\mathsf{Adv}_1$ and $\mathsf{Adv}_1\approx\mathsf{Adv}_2$, we have $\mathsf{Adv}_0\approx\mathsf{Adv}_2$3. Therefore, the adversary's advantage in the original protocol ($\mathsf{Game}_0$) is negligible. This proves that the Commit-Reveal protocol is unpredictable.
	\end{proof}
	
\subsection{Code Repository} \label{appendix:repository}

The complete source code for both the fully on-chain and hybrid implementations of the Commit-Reveal\(^2\) protocol, along with test scripts, utilities, and experimental configurations, can be accessed in the repository:  
\url{https://github.com/tokamak-network/Commit-Reveal2/tree/full-paper}.
The gas reports and consolidated series used to generate the manuscript’s figures and tables are included in the repository (see the \texttt{output/} folder).

\subsection{Merkle Root Construction Code} \label{appendix:merkle-root}

The Merkle Root computation in the Commit-Reveal\(^2\) protocol ensures deterministic results by arranging the leaves in the order of activated nodes. The following code snippet demonstrates the \texttt{createMerkleRoot} function, which iteratively hashes the leaves to compute the Merkle Root:

\begin{lstlisting}[caption=Merkle Root Construction, label={code:merkle-root}]
function _createMerkleRoot(bytes32[] memory leaves) internal pure returns (bytes32 r) {
        assembly ("memory-safe") {
            let leavesLenInBytes := shl(5, mload(leaves))
            let hashCountInBytes := sub(leavesLenInBytes, 0x20) // unchecked sub, check outside of this function
            let hashes := mload(0x40)
            mstore(hashes, hashCountInBytes)
            let hashDataPtr := add(hashes, 0x20)
            let leafDataPtr := add(leaves, 0x20)
            let leafPosInBytes
            let hashPosInBytes
            for { let i } lt(i, hashCountInBytes) { i := add(i, 0x20) } {
                switch lt(leafPosInBytes, leavesLenInBytes)
                case 1 {
                    mstore(0x00, mload(add(leafDataPtr, leafPosInBytes)))
                    leafPosInBytes := add(leafPosInBytes, 0x20)
                }
                default {
                    mstore(0x00, mload(add(hashDataPtr, hashPosInBytes)))
                    hashPosInBytes := add(hashPosInBytes, 0x20)
                }
                switch lt(leafPosInBytes, leavesLenInBytes)
                case 1 {
                    mstore(0x20, mload(add(leafDataPtr, leafPosInBytes)))
                    leafPosInBytes := add(leafPosInBytes, 0x20)
                }
                default {
                    mstore(0x20, mload(add(hashDataPtr, hashPosInBytes)))
                    hashPosInBytes := add(hashPosInBytes, 0x20)
                }
                mstore(add(hashDataPtr, i), keccak256(0x00, 0x40))
            }
            mstore(0x40, add(hashDataPtr, hashCountInBytes)) // update the free memory pointer
            r := mload(add(hashDataPtr, sub(hashCountInBytes, 0x20)))
        }
    }
\end{lstlisting}

\subsection{EIP-712 Typed Data Hash Construction Code} \label{appendix:typed-data-hash}

The following code illustrates the EIP-712 compliant hash construction used for securely signing structured data in the Commit-Reveal\(^2\) protocol:

\begin{lstlisting}[caption=EIP-712 Typed Data Hash Construction, label={code:typed-data-hash}]
bytes32 hashTypedDataV4 = keccak256(
    abi.encodePacked(
        hex"19_01",
        keccak256(
            abi.encode(
                keccak256(
                    "EIP712Domain(string name,string version,uint256 chainId,address verifyingContract)"
                ),
                keccak256(bytes("Commit Reveal2")),
                keccak256(bytes("1")),
                block.chainid,
                address(s_commitReveal2)
            )
        ),
        keccak256(
            abi.encode(
                keccak256(
                    "Message(uint256 round,uint256 trialNum,bytes32 cv)"
                ),
                Message({
                    round: round,
                    trialNum: trialNum,
                    cv: cv
                })
            )
        )
    )
);
\end{lstlisting}

\subsection{Measured Abnormal‑Path Routes}
\label{appendix:abnormal-paths}

\providecommand{\flowsep}{\nobreak\hspace{0.15em}$\rightarrow$\hspace{0.15em}\allowbreak}
\providecommand{\SMR}{\texttt{submit\allowbreak Merkle\allowbreak Root()}}
\providecommand{\GenR}{\texttt{generate\allowbreak Random\allowbreak Number()}}
\providecommand{\ReqS}{\texttt{request\allowbreak To\allowbreak Submit\allowbreak S()}}
\providecommand{\FailS}{\texttt{fail\allowbreak To\allowbreak Submit\allowbreak S()}}
\providecommand{\FailReqOrGen}{\texttt{fail\allowbreak To\allowbreak Request\allowbreak S\allowbreak Or\allowbreak Generate\allowbreak Random\allowbreak Number()}}
\providecommand{\DepAct}{\texttt{deposit\allowbreak And\allowbreak Activate()}}
\providecommand{\Resume}{\mbox{\texttt{resume()}}}
\providecommand{\SubS}{\texttt{submit\allowbreak S()}}

\paragraph{Baseline (fallback enabled, used in Fig.~\ref{fig:abnormal-full} / Table~\ref{tab:abnormal-cost})}
\noindent
\SMR\flowsep \GenR.

\paragraph{Participant‑withholding route (used)}
\noindent
\SMR\flowsep \ReqS\flowsep \FailS\flowsep \SMR\flowsep \GenR.\\
For $n{=}2$ only:
\SMR\flowsep \ReqS\flowsep \FailS\flowsep \DepAct\flowsep \Resume\flowsep \SMR\flowsep \GenR.

\paragraph{Leader‑withholding route (used)}
\noindent
\SMR\flowsep \FailReqOrGen\flowsep \Resume\flowsep \SMR\flowsep \GenR.

\paragraph{Variant not plotted (treated as griefing in Section~\ref{subsec:griefing})}
\noindent
\SMR\flowsep \ReqS\flowsep \SubS.